\documentclass[11pt,oneside,reqno]{article}
\usepackage[utf8]{inputenc}
\usepackage{braket}
\usepackage{graphicx}
\usepackage{amsmath}
\usepackage{amssymb}
\usepackage{amsthm}
\usepackage{marginnote}
\usepackage{xcolor}
\usepackage{hyperref}
\usepackage{tikz}
\usepackage{caption}
\usepackage{subcaption}
\usetikzlibrary{quantikz2}
\usepackage{float}
\usepackage[margin=1in]{geometry}
\hypersetup{
    colorlinks=true,
    linkcolor=blue,
    filecolor=magenta,      
    urlcolor=blue,
    citecolor=red,
    pdftitle={bradshaw},
    pdfpagemode=FullScreen,
    }
\newtheorem{theorem}{Theorem}
\newtheorem{lemma}{Lemma}

\newtheorem{proposition}{Proposition}

\DeclareMathOperator{\tr}{Tr}

\DeclareMathOperator{\re}{Re}

\title{A Closed Form for Moment-Based Entanglement Tests Associated to the PPT Criterion}
\date{$^\textnormal{a}$Advanced Processing Branch, Naval Surface Warfare Center, Panama City Division\\
$^\textnormal{b}$Advanced Acoustics \& Seabed Warfare Branch, Naval Surface Warfare Center, Panama City Division}
\begin{document}
\author{Zachary P. Bradshaw$^\textnormal{a}$\footnote{Corresponding Author: zbradshaw@tulane.edu} and Margarite L. LaBorde$^\textnormal{b}$}

\maketitle
\begin{abstract}
    Neven et al. have explored an unexpected alliance between the mathematical insights of Sir Isaac Newton and Ren\'e Descartes which culminates in the reduction of the Positive Partial Transpose (PPT) criterion to an equivalent hierarchy of entanglement tests based on the moments of the partial transpose. By repurposing these classical results in the context of modern quantum theory, they illuminate new pathways for entanglement verification. Here, we expand on this work by providing a closed form for the inequalities defining these entanglement tests and producing an equivalent set of graph theoretic conditions on the weighted graph induced by the partial transpose.
\end{abstract}

\section{Introduction}\label{sec:intro}

In recent decades, entanglement has been understood as a resource for performing computations that may not be feasible on a classical computer---although it has been shown that entanglement alone is not sufficient to outperform our trusty classical devices (see the Gottesman–Knill theorem \cite{gottesman1998}). Still, for this reason and many others, researchers have sought easily computable necessary and sufficient conditions for a quantum state to be entangled \cite{DPS02,doherty04separability,doherty05multipartite}. Unfortunately, it was shown that this problem is NP-hard\footnote{A problem $H$ is NP-hard if for any other problem in NP, there is a polynomial time reduction from $L$ to $H$. Thus, a polynomial time solution to $H$ would imply P=NP.} \cite{gharibian2008}, suggesting that such conditions do not exist; however, there are a variety of necessary conditions which one may make use of, and the most widely known is perhaps the Positive Partial Transpose (PPT) criterion \cite{ppt1,ppt3,ppt2} due to Peres \cite{Per96} and the Horodecki family \cite{HORODECKI1996}. This criterion states that when a quantum state is separable, the partial transpose of its density matrix representation has only nonnegative eigenvalues.

Neven et al. have shown that the PPT criterion can essentially be split into several weaker criteria which are collectively equivalent to PPT \cite{neven2021}, and they do so by combining Newton's recursive identity for the elementary symmetric polynomials \cite{Newton} with Descartes' rule of sign \cite{bensimhoun2016} to establish several inequalities between the moments of the partial transpose which are more efficiently computed. A similar approach is taken in \cite{yu2021} using Hankel matrices. The connection to Newton's identities parallels that in \cite{bradshaw2022}, where a similar recursion formula was used for the related homogeneous symmetric polynomials to justify the monotonicity of the entanglement tests therein.

In this work, we give a closed form for the family of entanglement tests given by Neven et al. We do so by deriving a closed form for Newton's recursive identities using generating functions and Fa\`a di Bruno's formula. As a consequence, we also prove a formula appearing in a problem posed by T. Amdeberhan \cite{amdeberhan2022} regarding the determinant of a matrix in terms of its moments. This problem was recently addressed by Zhan and Huang \cite{zhan2025} using combinatorial methods. Additionally, we translate these entanglement tests into the graph zeta function picture introduced in \cite{bradshaw2023}, deriving an equivalent set of graph theoretic conditions on the weighted graph induced by the partial transpose of the matrix in question.

The remainder of this article is organized as follows. In Section~\ref{sec:newton}, we derive Newton's recursive identity for the elementary symmetric polynomials as well as Descartes' rule of sign. In Section~\ref{sec:weakppt}, we review the construction of the moment-based entanglement tests of Neven et al., and in Section~\ref{sec:closed-form}, we expand on their work by connecting these results to the complete exponential Bell polynomials and the closely related cycle index polynomials of the symmetric group using a generating function argument for the elementary symmetric polynomials. This connection provides a closed form for the inequalities used by Neven et al. We then show in Section~\ref{sec:implementation} that it does not suffice to check only the final inequality by producing an example of an entangled function for which the inequalities are satisfied after violating a previous step. In Section~\ref{sec:graphs}, we translate these tests to the theory of graph zeta functions and produce the equivalent graph theoretic conditions for the PPT criterion. Finally, in Section~\ref{sec:conclusion}, we give concluding remarks.

\section{Newton's Identities and Descartes' Rule of Sign}\label{sec:newton}
The $k$-th elementary symmetric polynomial in $n$ variables \cite{macdonald1995} is defined by
\begin{equation}
    e_k(x_1,\ldots,x_n) = \sum_{1\le j_1<\cdots<j_k\le n}x_{j_1}\cdots x_{j_k},
\end{equation}
with $e_k(x_1,\ldots,x_n) = 0$ when $k>n$. These polynomials generate the ring of symmetric polynomials consisting of all polynomials $p(x_1,\ldots,x_n)$ with the property that $p(x_{\sigma(1)},\ldots,x_{\sigma(n)})=p(x_1,\ldots,x_n)$ for every permutation $\sigma$ in the symmetric group $S_n$ on $n$ letters. That is, the ring of polynomials that are invariant under the natural action of $S_n$. These polynomials appear in a variety of contexts, but of particular importance to us is their appearance in Vieta's formula, which relates the roots of a monic polynomial with its coefficients. Explicitly, we have
\begin{equation}\label{eq:vieta}
    \prod_{j=1}^n(t-x_j) = \sum_{k=0}^n(-1)^ke_k(x_1,\ldots,x_n)t^{n-k},
\end{equation}
and from this identity, we can derive a recursive relationship between the elementary symmetric polynomials and the power sum symmetric polynomials defined by $p_k(x_1,\ldots,x_n) = \sum_{j=1}^nx_j^k$. To do so, we work in the field of fractions of the ring of formal power series with integer coefficients; although, Mead finds this approach unsatisfactory and produces an alternative derivation facilitated by a notational change in \cite{Newton}, which the reader may find useful.

\begin{lemma}[Newton's Identities]\label{lemma:newton}
    Let $n\ge k\ge1$. Then
    \begin{equation}
        ke_k(x_1,\ldots,x_n) = \sum_{j=1}^k(-1)^{j-1}e_{k-j}(x_1,\ldots,x_n)p_i(x_1,\ldots,x_n).
    \end{equation}
    Moreover, for $k>n\ge1$, we have
    \begin{equation}
        \sum_{j=k-n}^k(-1)^{j-1}e_{k-j}(x_1,\ldots,x_n)p_i(x_1,\ldots,x_n)=0
    \end{equation}
\end{lemma}
\begin{proof}
    Substituting $t\to1/t$ in Vieta's formula \eqref{eq:vieta} produces
    \begin{equation}
        \prod_{j=1}^n\left(\frac{1}{t}-x_j\right) = \sum_{k=0}^n(-1)^ke_k(x_1,\ldots,x_n)t^{k-n}
    \end{equation}
    and now multiplying by $t^n$ gives
    \begin{equation}\label{eq:vieta-sub}
        \prod_{j=1}^n(1-x_jt) = \sum_{k=0}^n(-1)^ke_k(x_1,\ldots,x_n)t^{k}.
    \end{equation}
    Differentiating with respect to $t$ produces
    \begin{equation}
        \sum_{k=1}^nke_k(x_1,\ldots,x_n)t^{k-1}=-\sum_{i=1}^nx_i\prod_{j\ne i}(1-x_jt),
    \end{equation}
    so that upon multiplying by $t$, rewriting the right hand side, and making use of the geometric series expansion, we are left with
    \begin{align}
        \sum_{k=1}^nke_k(x_1,\ldots,x_n)t^{k}&=-\sum_{i=1}^n\frac{x_it}{1-x_it}\prod_{j=1}^n(1-x_jt)\\
        &=-\sum_{i=1}^n\sum_{\ell=1}^\infty x_i^\ell t^\ell\prod_{j=1}^n(1-x_jt)\\
        &=\sum_{\ell=1}^\infty p_\ell(x_1,\ldots,x_n) t^\ell\sum_{j=0}^n(-1)^{j-1}e_j(x_1,\ldots,x_n)t^{j},
    \end{align}
    where in the last equality, we have made use of \eqref{eq:vieta-sub}. Comparing the coefficients of $t^k$ on either side completes the proof.
\end{proof}

This relationship appears throughout mathematics in areas such as Galois theory and combinatorics, and in a moment, we will further showcase its use in the detection of quantum entanglement. We now shift our attention to our next essential element, Descartes' rule of sign. This rule appears in Descartes' 1637 work \textit{La G\'eom\'etrie}, which is itself an appendix to his earlier work \textit{Discours de la m\'ethode} \cite{descartes2006discourse}, where he lays out his method for discerning truth in the sciences. By examining the changes in the signs of the coefficients of a polynomial, Descartes is able to give a bound for the number of positive roots which the polynomial may have. Explicitly, the rule of sign is as follows.

\begin{lemma}[Descartes' Rule of Sign] Let $p(x) = a_nx^n+\cdots+a_1x+a_0$ be a polynomial with real coefficients. The number of positive roots of $p$ is bounded above by the number of sign changes between consecutive coefficients.
\end{lemma}
\begin{proof}
    Let $r$ denote the number of positive roots and let $s$ denote the number of sign changes. Note that if $a_0=0$, we can divide out a factor of $x$ without changing the number of \textit{positive} roots. We may therefore assume that $a_0\ne0$. Consider the quantity given by the product of the first and last coefficients. If $a_na_0>0$, then $r$ must be even. Indeed, $a_0$ determines the sign of $p$ at $x=0$ while $a_n$ determines the sign of $p$ as $x\to\infty$. Since $a_na_0>0$ implies that $a_n$ and $a_0$ share the same sign, it follows that $p$ crosses the positive $x$-axis an even number of times (each of which contributes an odd multiplicity). The polynomial $p$ is also allowed to touch the positive $x$-axis without crossing it, but each such zero contributes an even multiplicity to the overall count. It follows that $r$ must be even. Similarly, if $a_na_0<0$, then $r$ must be odd. Thus, $r$ and $s$ always have the same parity.

    When $n=0$ or $n=1$, it is clear that the number of positive roots is bounded above by the number of sign changes between consecutive coefficients (and that they differ by an even number). Let us proceed by induction on $n$. Suppose the lemma is true for some $n-1\ge2$. Taking the derivative of $p$ produces $p'(x)=a_nx^{n-1}+\cdots+a_1$, which is a polynomial of degree $n-1$. Then by the induction hypothesis, there is an integer $m\ge0$ such that $s'-r'=2m$, where we have denoted the number of positive roots of $p'$ by $r'$ and the number of sign changes by $s'$. By Rolle's theorem, the derivative $p'$ has a root between any two roots of the polynomial $p$. Moreover, any root of $p$ of multiplicity $k$ is also a root of $p'$ of multiplicity $k-1$, as can be seen by writing $p$ as a product of linear terms and differentiating using the product rule. It follows that $r'\ge r-1$.

    Now, since every exponent is positive, if $a_0a_1>0$, then $s'=s$. Otherwise, $s'=s-1$. Thus, we finally have $r\le r+1=s'-2m+1\le s-2m+1\le s+1$. But $r$ and $s$ share the same parity, so we actually have $r\le s$, and this completes the proof.
\end{proof}

\section{Entanglement Criteria}\label{sec:weakppt}

Before reviewing the entanglement criteria of Neven et al., let us recall the definition of the partial transpose operation. Given a composite quantum system described by the density operator $\rho$, we can construct the partial transpose of $\rho$ in the following way. Since the system is composite, the Hilbert space decomposes into a tensor product, say $\mathcal{H}_A\otimes\mathcal{H}_B$, and the density operator, which acts on this space, can therefore be written as 
\begin{equation}
    \rho = \sum_{i,j,k,l}\rho^{ik}_{jl}\ket{i}\!\!\bra{j}\otimes\ket{k}\!\!\bra{l},
\end{equation}
where the vectors $\{\ket{i}\otimes\ket{k}\}$ form a basis for the composite space and the $c_{ijkl}$ are complex coefficients. The partial transpose of $\rho$ with respect to the subsystem $B$ is defined by
\begin{equation}
    \rho^{T_B}= \sum_{i,j,k,l}\rho^{ik}_{jl}\ket{i}\!\!\bra{j}\otimes\ket{l}\!\!\bra{k}=\sum_{i,j,k,l}\rho^{il}_{jk}\ket{i}\!\!\bra{j}\otimes\ket{k}\!\!\bra{l},
\end{equation}
and the partial transpose with respect to subsystem $A$ is defined similarly. 

A well established method to test for entanglement in a quantum state is the so-called PPT criterion due to Peres \cite{Per96} and the Horodecki family \cite{HORODECKI1996}, which simply states that the partial transpose of $\rho$ with respect to any subsystem is positive-semidefinite whenever $\rho$ is separable. Thus, the existence of a negative eigenvalue of some partial transpose implies that $\rho$ is entangled. This condition is in general only necessary, though it is also sufficient for the special cases $\dim(\mathcal{H}_A)=2=\dim(\mathcal{H}_B)$ and $\dim(\mathcal{H}_A)=2, \dim(\mathcal{H}_B)=3$.

Neven et al. have come up with a family of weaker criteria which are collectively equivalent to PPT and are each more efficiently computed than the PPT criterion itself. The conditions follow from Newton's identities and the next lemma, which is proven with an application of Descartes' rule of sign to the characteristic polynomial of a positive semi-definite matrix written in terms of the elementary symmetric polynomials. 

\begin{lemma}\label{lem:pos-sym}
Let $A$ denote a self-adjoint matrix acting on a Hilbert space of dimension $d$. Then $A$ is positive semi-definite if and only if $e_i(\lambda_1,\ldots,\lambda_d)\ge0$ for all $i=1,\ldots,d$, where $\lambda_1,\ldots,\lambda_d$ are the eigenvalues of $A$, and $e_i$ denotes the $i$-th elementary symmetric polynomial.
\end{lemma}
\begin{proof} On the one hand, if $A$ is positive semi-definite, then all of its eigenvalues are nonnegative and it follows that $e_i(\lambda_1,\ldots,\lambda_d)\ge0$ immediately from its definition. Conversely, suppose $e_i(\lambda_1,\ldots,\lambda_d)\ge0$ for all $i=1,\ldots,d$. Let $P(t)=\prod_{i=1}^d(\lambda_i-t)$ be the characteristic polynomial of A and observe that
\begin{equation}P(-t)=\prod_{i=1}^d(\lambda_i+t)=\sum_{i=1}^de_i(\lambda_1,\ldots,\lambda_d)t^{d-i}.
\end{equation}
Now, $A$ is positive semi-definite if and only if $P(t)$ has only positive roots, which is true if and only if $P(-t)$ has only negative roots. Then by Descartes' rule of sign, it follows that $A$ is positive semi-definite.
\end{proof}

For brevity, let us write Newton's identities as
\begin{equation}\label{eq:newton}
ke_k=\sum_{i=1}^k(-1)^{i-1}e_{k-i}p_i,
\end{equation}
where it is understood that the polynomials are evaluated at the same arguments. If $\lambda_1,\ldots,\lambda_d$ are the eigenvalues of a self-adjoint operator $A$, notice that $p_i(\lambda_1,\ldots,\lambda_d)=\tr(A^i)$ is the $i$-th moment of $A$. Letting $A=\rho^\Gamma$ be the partial transpose of some density operator $\rho$ with respect to some subsystem, we invoke Lemma~\ref{lem:pos-sym}, producing inequalities between the moments of the partial transpose whenever it is positive semi-definite. In fact, the satisfaction of every such inequality is equivalent to the partial transpose being positive semi-definite by the lemma. If any such inequality fails, then the state $\rho$ is entangled, as it will fail the PPT test.

Let us take a look at some of these inequalities. To clean up the notation, we will write $p_j=\tr((\rho^\Gamma)^j)$. From \eqref{eq:newton} and the assumption that $e_k\ge0$, we have
\begin{equation}\label{eq:first-ineq}
p_1\ge0
\end{equation}
\begin{equation}
p_2\le p_1^2
\end{equation}
\begin{equation}\label{eq:first-nontrivial}
p_3\ge -\frac12 p_1^3+\frac32 p_1p_2
\end{equation}
\begin{equation}\label{eq:last-ineq}
p_4\le \frac12(p_1^2-p_2)^2-\frac13 p_1^4+\frac43 p_1p_3.
\end{equation}
As noted in \cite{neven2021}, the first inequality is satisfied trivially since $\tr(\rho^\Gamma)=1$ whenever $\rho$ is a density matrix. Likewise, the second inequality is trivially satisfied since $\tr((\rho^\Gamma)^2)=\tr(\rho^2)\le1$, so that the first non-trivial inequality is \eqref{eq:first-nontrivial}, which requires an estimate of only the second and third moments of $\rho^\Gamma$ to test.

\section{Closed Form Derivation}\label{sec:closed-form}
Let us now expand on the work of Neven et al. The generalization of the inequalities \eqref{eq:first-ineq}-\eqref{eq:last-ineq} defined by Newton's identities can be given a closed form and are related to the Bell polynomials, as well as the cycle index polynomials of the symmetric and alternating groups. We will need the following lemma.
\begin{lemma}\label{lem:gen}
    The ordinary generating function for the elementary symmetric polynomials is given by
    \begin{equation}\label{eq:gen}
        \sum_{k=0}^\infty e_kt^k = \exp\left(\sum_{k=1}^\infty\frac{(-1)^{k+1}}{k}p_kt^k\right).
    \end{equation}
\end{lemma}
\begin{proof}
    From Vieta's formula, we have
    \begin{equation}
        \sum_{k=0}^\infty e_kt^k = \prod_{k=1}^n(1+\lambda_kt),
    \end{equation}
    where it is again understood that $e_k=e_k(\lambda_1,\ldots,\lambda_n)$. Then by expanding the power sum in \eqref{eq:gen}, we have
    \begin{align}
        \exp\left(\sum_{k=1}^\infty\frac{(-1)^{k+1}}{k}p_kt^k\right)&=\exp\left(\sum_{k=1}^\infty\frac{(-1)^{k+1}}{k}\sum_{j=1}^n\lambda_j^{k}t^k\right)\\
        &=\prod_{j=1}^n\exp\left(\sum_{k=1}^\infty\frac{(-1)^{k+1}}{k}\lambda_j^kt^k\right)\\
        &=\prod_{j=1}^n\exp\left(\log(1+\lambda_jt)\right)\\
        &=\prod_{j=1}^n(1+\lambda_jt).
    \end{align}
    This completes the proof.
\end{proof}

The inequalities given by Newton's identities are now given by taking derivatives of the generating function and evaluating at $t=0$. This can be accomplished using Fa\`a di Bruno's formula for an arbitrary derivative of a composition of functions \cite{bruno}. Explicitly, we have the following proposition.

\begin{proposition}
    The elementary symmetric polynomials are given by 
    \begin{equation}\label{eq:elem-sum}
    e_k = (-1)^k\sum_{n_1+2n_2+\cdots+kn_k=k}\prod_{j=1}^k\frac{(-p_j)^{n_j}}{n_j!j^{n_j}},
\end{equation}
where the sum is over all choices of $n_1,\ldots,n_k$ such that $n_1+2n_2+\cdots+kn_k=k$.
\end{proposition}
\begin{proof}
    From Lemma~\ref{lem:gen}, we have \eqref{eq:gen}. Taking the $m$-th derivative of the left hand side and evaluating at $t=0$ yields $m!e_m$. Doing the same on the right hand side is accomplished by Fa\`a di Bruno's formula. We have
    \begin{align}
        m!e_m&=\frac{d^m}{dt^m}\bigg\vert_{t=0}\left(\exp\left(\sum_{k=1}^\infty\frac{(-1)^{k+1}}{k}p_kt^k\right)\right)\\
        &=\sum_{n_1+2n_2+\cdots+mn_m=m}\prod_{j=1}^m\frac{m!}{n_j!(j!)^{n_j}}\left(\frac{d^j}{dt^j}\bigg\vert_{t=0}\left(\sum_{k=1}^\infty\frac{(-1)^{k+1}}{k}p_kt^k\right)\right)^{n_j}\\
        &=\sum_{n_1+2n_2+\cdots+mn_m=m}\prod_{j=1}^m\frac{m!}{n_j!(j!)^{n_j}}\left(\frac{(-1)^{j+1}j!}{j}p_j\right)^{n_j}\\
        &=(-1)^m\sum_{n_1+2n_2+\cdots+mn_m=m}\prod_{j=1}^m\frac{m!}{n_j!j^{n_j}}(-p_j)^{n_j},
    \end{align}
    and this completes the proof.
\end{proof}

Remarkably, this quantity is linked to the Bell polynomials \cite{bell1934exponential}, a combinatorial object appearing in the theory of set partitions. Indeed, the $k$-th complete exponential Bell polynomial is defined by $B_k(x_1,\ldots,x_k)$
\begin{equation}
    B_k(x_1,\ldots,x_k)=k!\sum_{n_1+2n_2+\cdots+kn_k=k}\prod_{j=1}^k\frac{x_j^{n_j}}{n_j!j^{n_j}},
\end{equation}
and so it follows that \eqref{eq:elem-sum} can be written
\begin{equation}
    e_k = \frac{(-1)^k}{k!}B_k(-p_1,-1!p_2,\ldots,-(k-1)!p_k).
\end{equation}
 
Moreover, this quantity is further linked to the closely related cycle index polynomial of a finite permutation group $G$ \cite{brualdi2010,tucker1995}, a polynomial in several variables designed to encode information about the structure of permutations in the group. Explicitly, it is defined as 
\begin{equation}
    Z(G)(x_1,\ldots,x_n)=\frac{1}{\lvert G\rvert}\sum_{\sigma\in G}x_1^{j_1(\sigma)}\cdots x_n^{j_n(\sigma)},
\end{equation}
where $j_k(\sigma)$ is the number of cycles of length $k$ in the standard cycle decomposition of $\sigma$ and $\lvert G\rvert$ denotes the order of the group. This combinatorial object is an important piece of P\'olya theory \cite{polya,combinatorics} and has previously appeared in a quantum information setting in \cite{bradshaw2022}, where it was shown that for every finite group there is a test for entanglement with an acceptance probability given by the cycle index polynomial of the group. It subsequently appeared in \cite{bradshaw2023}, where these tests were given a graph-theoretic interpretation by defining a graph zeta function for quantum states. It can be shown that the cycle index polynomial of the symmetric group is $Z(S_n)(x_1,\ldots,x_n)=\frac{1}{n!}B_n(x_1,1!x_2,\ldots,(n-1)!x_n)$, and so we conclude also that
\begin{equation}
    e_k = (-1)^k Z(S_k)(-p_1,\ldots,-p_k).
\end{equation}
Noting once more that $p_j(\lambda_1,\ldots,\lambda_n) = \tr((\rho^\Gamma)^j)$, where the arguments $\lambda_i$ are the eigenvalues of $\rho^\Gamma$, the observations of this section culminate in the following theorem.

\begin{theorem}\label{thm:explicit}
    Let $\rho$ be a density operator. If $\rho$ is not entangled, then
    \begin{equation}\label{eq:closedform}
        (-1)^k Z(S_k)(-\tr[\rho^\Gamma],-\tr[(\rho^\Gamma)^2],\ldots,-\tr[(\rho^\Gamma)^k])\ge0
    \end{equation}
    for some value of $k$. Written explicitly, if $\rho$ is not entangled, then
    \begin{equation}
        (-1)^k\sum_{n_1+2n_2+\cdots+kn_k=k}\prod_{j=1}^k\frac{(-\tr[(\rho^\Gamma)^j])^{n_j}}{n_j!j^{n_j}}\ge0
    \end{equation}
    for some value of $k$.
\end{theorem}

Note that the cycle index polynomial of the alternating group $A_k$ is
\begin{equation}
    Z(A_k)(x_1,\ldots,x_k)=\sum_{n_1+2n_2+\cdots+kn_k=k}(1+(-1)^{n_2+n_4+\cdots})\prod_{j=1}^k\frac{x_j^{n_j}}{n_j!j^{n_j}},
\end{equation}
and this allows us to write \eqref{eq:closedform} equivalently as
\begin{equation}
    Z(A_k)(\tr[\rho^\Gamma],\ldots,\tr[(\rho^\Gamma)^k])\ge Z(S_k)(\tr[\rho^\Gamma],\ldots,\tr[(\rho^\Gamma)^k]).
\end{equation}
Expanding using the definition of the cycle index polynomial produces
\begin{equation}
    \frac{2}{k!}\sum_{\sigma\in A_k}\prod_{j=1}^k(\tr[(\rho^\Gamma)^j])^{c_j(\sigma)}\ge\frac{1}{k!}\sum_{\sigma\in S_k}\prod_{j=1}^k(\tr[(\rho^\Gamma)^j])^{c_j(\sigma)},
\end{equation}
but the symmetric group contains the alternating group, and so we have
\begin{equation}
    \sum_{\sigma\in A_k}\prod_{j=1}^k(\tr[(\rho^\Gamma)^j])^{c_j(\sigma)}\ge\sum_{\sigma\in S_k\setminus A_k}\prod_{j=1}^k(\tr[(\rho^\Gamma)^j])^{c_j(\sigma)}.
\end{equation}
Thus, we see that if we compute the product $\prod_{j=1}^k(\tr[(\rho^\Gamma)^j])^{c_j(\sigma)}$ for all $\sigma\in S_k$ and separately sum the contributions from even and odd permutations, the inequality is violated precisely when the sum over the odd permutations is larger than the sum over the even permutations.

\section{Implementation on Quantum Computers}\label{sec:implementation}
\sloppy
With Theorem~\ref{thm:explicit} in hand, let us examine these inequalities more closely. Let $f(k)=(-1)^kZ(S_k)(-\tr[\rho^\Gamma],\ldots,-\tr[(\rho^\Gamma)^k])$ so that the inequality is written $f(k)\ge0$. In what follows, we use the notation $\sigma(\ket{i_1}\otimes\cdots\otimes\ket{i_n})=\ket{i_{\sigma^{-1}(1)}}\otimes\cdots\otimes\ket{i_{\sigma^{-1}(n)}}$ to denote the natural action of a permutation group on a tensor product space.

\begin{lemma}\label{lemma:moment-computation}
    Let $\sigma_j = (j\ \ j-1\ \ j-2\ \cdots\ 1)$ be a cyclic permutation of order $j$, and let $\rho$ be a density matrix. Then
    \begin{equation}
        \tr[(\rho^\Gamma)^n]=\tr[(\sigma_n\otimes\sigma_n^{-1})\rho^{\otimes n}].
    \end{equation}
\end{lemma}
\begin{proof}
    Note that the inverse of $\sigma_j$ is $\sigma_j^{-1}=(1\ 2\ 3\ \cdots\ j)$. If we label the coefficients of $\rho$ by $\rho^{ik}_{jl}$ so that
    \begin{equation}
        \rho = \sum_{i,j,k,l}\rho^{ik}_{jl}\ket{i}\!\!\bra{j}\otimes\ket{k}\!\!\bra{l},
    \end{equation}
    then $\rho^\Gamma$ is defined by
    \begin{equation}
        \rho^\Gamma = \sum_{i,j,k,l}\rho^{ik}_{jl}\ket{i}\!\!\bra{j}\otimes\ket{l}\!\!\bra{k},
    \end{equation}
    and the $n$-th power of $\rho^\Gamma$ is
    \begin{align}
        \sum&\rho^{i_1k_1}_{j_1l_1}\rho^{i_2k_2}_{j_2l_2}\cdots\rho^{i_nk_n}_{j_nl_n}\ket{i_1}\!\!\braket{j_1\vert i_2}\cdots\braket{j_{n-1}\vert i_n}\!\!\bra{j_n}\otimes\ket{l_1}\!\!\braket{k_1\vert l_2}\cdots\braket{k_{n-1}\vert l_n}\!\!\bra{k_n}\\
        &=\sum\rho^{i_1l_2}_{j_1l_1}\rho^{j_1l_3}_{j_2l_2}\cdots\rho^{j_{n-1}k_n}_{j_nl_n}\ket{i_1}\!\!\bra{j_n}\otimes\ket{l_1}\!\!\bra{k_n}.
    \end{align}
    Thus, it follows that 
    \begin{equation}
        \tr[(\rho^\Gamma)^n]=\sum\rho^{j_nl_2}_{j_1l_1}\rho^{j_1l_3}_{j_2l_2}\cdots\rho^{j_{n-1}l_1}_{j_nl_n}.
    \end{equation}
    On the other hand, $\rho^{\otimes n}$ is given by
    \begin{equation}
        \rho^{\otimes n}=\sum\rho^{i_1k_1}_{j_1l_1}\rho^{i_2k_2}_{j_2l_2}\cdots\rho^{i_nk_n}_{j_nl_n}\ket{i_1}\!\!\bra{j_1}\otimes\ket{k_1}\!\!\bra{l_1}\otimes\cdots\otimes\ket{i_n}\!\!\bra{j_n}\otimes\ket{k_n}\!\!\bra{l_n},
    \end{equation}
    and so $(\sigma_n\otimes\sigma_n^{-1})\rho^{\otimes n}$ is
    \begin{equation}
        (\sigma_n\otimes\sigma_n^{-1})\rho^{\otimes n}=\sum\rho^{i_1k_1}_{j_1l_1}\rho^{i_2k_2}_{j_2l_2}\cdots\rho^{i_nk_n}_{j_nl_n}\ket{i_2}\!\!\bra{j_1}\otimes\ket{k_n}\!\!\bra{l_1}\otimes\cdots\otimes\ket{i_1}\!\!\bra{j_n}\otimes\ket{k_{n-1}}\!\!\bra{l_n},
    \end{equation}
    and it follows that
    \begin{equation}
        \tr[(\sigma_n\otimes\sigma_n^{-1})\rho^{\otimes n}]=\sum\rho^{j_nl_2}_{j_1l_1}\rho^{j_1l_3}_{j_2l_2}\cdots\rho^{j_{n-1}l_1}_{j_nl_n}.
    \end{equation}
    Thus, we have $\tr[(\rho^\Gamma)^n]=\tr[(\sigma_n\otimes\sigma_n^{-1})\rho^{\otimes n}]$.

\end{proof}

Since $\rho^\Gamma$ may not be a density matrix itself, we cannot expect to use it as the input to some quantum algorithm which computes its moments. Fortunately, Lemma~\ref{lemma:moment-computation} tells us that the moments can be computed instead as the trace of a permutation operator applied to a tensor power of $\rho$. This equivalence allows us to construct the circuit shown in Figure~\ref{fig:momentestimation} for estimating the moments of $\rho^\Gamma$ as outlined in \cite{carteret2005,leifer2004}. Here, the $\rho^{\otimes k}$ input state is more easily understood in the ensemble of pure states picture. Indeed, the density matrix is a mathematical tool to express an ensemble $\{(p_1,\ket{\psi_1}\ket{\phi_1}),\ldots,(p_m,\ket{\psi_m}\ket{\phi_m})\}$ of pairs $(p_i,\ket{\psi_i}\ket{\phi_i})$ with $p_1+\cdots+p_m=1$, which indicates that the state of our system is $\ket{\psi_i}\ket{\phi_i}$ with probability $p_i$. This probabilistic picture reflects some classical ignorance about which quantum state the system is in, as opposed to the non-determinism induced by the Born rule, a quantum mechanical effect. Thus, the states being passed through the circuit are probabilistically determined by the ensemble; the state is $\ket{\psi_{i_1}}\ket{\phi_{i_1}}\cdots\ket{\psi_{i_k}}\ket{\phi_{i_k}}$ with probability $p_{i_1}\cdots p_{i_k}$.

\begin{figure}
    \centering
    \begin{quantikz}
        \lstick{$\ket{0}$} & \gate[1]{H} & \ctrl{1}           &                  \gate[1]{H} &\meter{}\\
        \lstick[2]{$\rho$} &             & \gate[5]{\sigma_k\otimes\sigma_k^{-1}} &  &             \\
                           &             &                    &                         &             \\
        \lstick{\vdots\ \ }\setwiretype{n}    &             &                    &                                      &\\
        \lstick[2]{$\rho$} &             &                    &                         &             \\
                           &             &                    &                         &             
    \end{quantikz}
    \caption{
        Circuit for estimating $\tr[(\rho^\Gamma)^k]$.
    }\label{fig:momentestimation}
\end{figure}
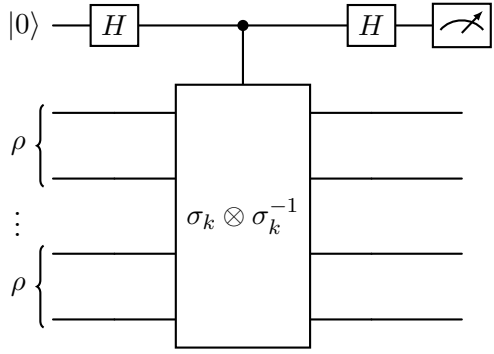

Given that the state of the system is $\ket{\psi_{i_1}}\ket{\phi_{i_1}}\cdots\ket{\psi_{i_k}}\ket{\phi_{i_k}}$, we can now perform an analysis of the circuit as usual. The state just before measurement is
\begin{equation}
    \ket{0}\frac{1+\sigma_k\otimes\sigma_k^{-1}}{2}\ket{\psi_{i_1}}\ket{\phi_{i_1}}\cdots\ket{\psi_{i_k}}\ket{\phi_{i_k}}+\ket{1}\frac{1-\sigma_k\otimes\sigma_k^{-1}}{2}\ket{\psi_{i_1}}\ket{\phi_{i_1}}\cdots\ket{\psi_{i_k}}\ket{\phi_{i_k}},
\end{equation}
and the probability of measuring $\ket{0}$ is therefore
\begin{align}
    P(\ket{0})&=\bra{\phi_{i_k}}\bra{\psi_{i_k}}\cdots\bra{\phi_{i_1}}\bra{\psi_{i_1}}\left(\frac{1+\sigma_k\otimes\sigma_k^{-1}}{2}\right)^\dagger\left(\frac{1+\sigma_k\otimes\sigma_k^{-1}}{2}\right)\ket{\psi_{i_1}}\ket{\phi_{i_1}}\cdots\ket{\psi_{i_k}}\ket{\phi_{i_k}}\\
    &=\bra{\phi_{i_k}}\bra{\psi_{i_k}}\cdots\bra{\phi_{i_1}}\bra{\psi_{i_1}}\left(\frac{1+\sigma_k^{-1}\otimes\sigma_k}{2}\right)\left(\frac{1+\sigma_k\otimes\sigma_k^{-1}}{2}\right)\ket{\psi_{i_1}}\ket{\phi_{i_1}}\cdots\ket{\psi_{i_k}}\ket{\phi_{i_k}}\\
    &=\bra{\phi_{i_k}}\bra{\psi_{i_k}}\cdots\bra{\phi_{i_1}}\bra{\psi_{i_1}}\frac{2+\sigma_k^{-1}\otimes\sigma_k+\sigma_k\otimes\sigma_k^{-1}}{4}\ket{\psi_{i_1}}\ket{\phi_{i_1}}\cdots\ket{\psi_{i_k}}\ket{\phi_{i_k}}.
\end{align}
Similarly, the probability of measuring $\ket{1}$ is
\begin{equation}
    P(\ket{1})=\bra{\phi_{i_k}}\bra{\psi_{i_k}}\cdots\bra{\phi_{i_1}}\bra{\psi_{i_1}}\frac{2-\sigma_k^{-1}\otimes\sigma_k-\sigma_k\otimes\sigma_k^{-1}}{4}\ket{\psi_{i_1}}\ket{\phi_{i_1}}\cdots\ket{\psi_{i_k}}\ket{\phi_{i_k}},
\end{equation}
and so the expected value of the circuit for this choice of input is
\begin{align}
    \langle Z\rangle_{i_1,\ldots,i_k} &= P(\ket{0})-P(\ket{1})\\
    &=\bra{\phi_{i_k}}\bra{\psi_{i_k}}\cdots\bra{\phi_{i_1}}\bra{\psi_{i_1}}\frac{\sigma_k^{-1}\otimes\sigma_k+\sigma_k\otimes\sigma_k^{-1}}{2}\ket{\psi_{i_1}}\ket{\phi_{i_1}}\cdots\ket{\psi_{i_k}}\ket{\phi_{i_k}}\\
    &=\re\left(\bra{\phi_{i_k}}\bra{\psi_{i_k}}\cdots\bra{\phi_{i_1}}\bra{\psi_{i_1}}\sigma_k\otimes\sigma_k^{-1}\ket{\psi_{i_1}}\ket{\phi_{i_1}}\cdots\ket{\psi_{i_k}}\ket{\phi_{i_k}}\right)\\
    &=\re\left(\tr[\sigma_k\otimes\sigma_k^{-1}\ket{\psi_{i_1}}\!\!\bra{\psi_{i_1}}\otimes\ket{\phi_{i_1}}\!\!\bra{\phi_{i_1}}\otimes\cdots\otimes\ket{\psi_{i_k}}\!\!\bra{\psi_{i_k}}\otimes\ket{\phi_{i_k}}\!\!\bra{\phi_{i_k}}]\right).
\end{align}
Thus, if we prepare the same mixed state for each run of the circuit, we see that $\ket{\psi_{i_1}}\ket{\phi_{i_1}}\cdots\ket{\psi_{i_k}}\ket{\phi_{i_k}}$ occurs with probability $p_{i_1}\cdots p_{i_k}$ and so the expected value for the circuit with $\rho^{\otimes k}$ as input is
\begin{align}
    \langle Z\rangle_\rho &= \sum_{i_1,\ldots,i_k=1}^mp_{i_1}\cdots p_{i_k}\langle Z\rangle_{i_1,\ldots,i_k}\\
    &=\sum_{i_1,\ldots,i_k=1}^mp_{i_1}\cdots p_{i_k}\re\left(\tr[\sigma_k\otimes\sigma_k^{-1}\ket{\psi_{i_1}}\!\!\bra{\psi_{i_1}}\otimes\ket{\phi_{i_1}}\!\!\bra{\phi_{i_1}}\otimes\cdots\otimes\ket{\psi_{i_k}}\!\!\bra{\psi_{i_k}}\otimes\ket{\phi_{i_k}}\!\!\bra{\phi_{i_k}}]\right)\\
    &=\re\left(\tr\left[\sigma_k\otimes\sigma_k^{-1}\sum_{i_1,\ldots,i_k=1}^mp_{i_1}\cdots p_{i_k}\ket{\psi_{i_1}}\!\!\bra{\psi_{i_1}}\otimes\ket{\phi_{i_1}}\!\!\bra{\phi_{i_1}}\otimes\cdots\otimes\ket{\psi_{i_k}}\!\!\bra{\psi_{i_k}}\otimes\ket{\phi_{i_k}}\!\!\bra{\phi_{i_k}}\right]\right)\\
    &=\re\left(\tr\left[\sigma_k\otimes\sigma_k^{-1}\rho^{\otimes k}\right]\right).
\end{align}
Since $\tr[\sigma_k\otimes\sigma_k^{-1}\rho^{\otimes k}]=\tr[(\rho^\Gamma)^k]$ and $\rho^{\Gamma}$ is self-adjoint, the moments must be real, and so the expected value of the circuit is 
\begin{equation}
    \langle Z\rangle_\rho = \tr[\sigma_k\otimes\sigma_k^{-1}\rho^{\otimes k}].
\end{equation}

\begin{figure}
\begin{center}
\begin{subfigure}[t]{0.45\textwidth}
         \centering
         \includegraphics[width=\textwidth]{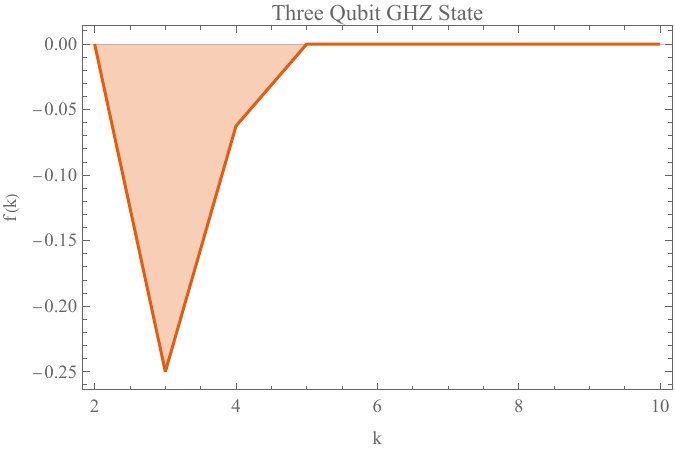}
         \caption{}
         \label{fig:GHZ}
\end{subfigure}
\begin{subfigure}[t]{0.45\textwidth}
         \centering
         \includegraphics[width=\textwidth]{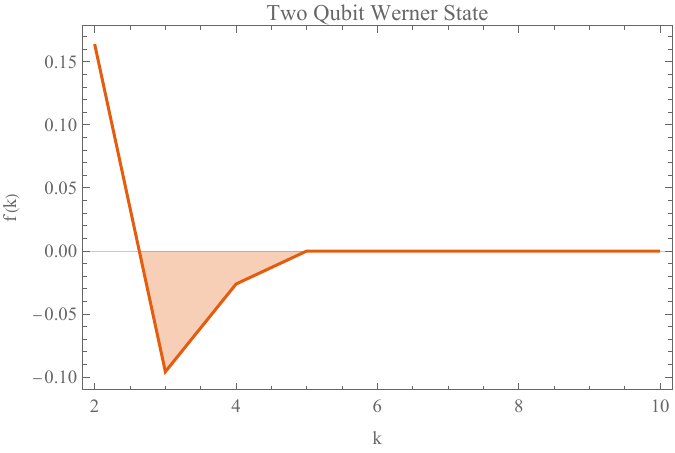}
         \caption{}
         \label{fig:Werner}
\end{subfigure}
\end{center}
\caption{A calculation of Eq.~\eqref{eq:closedform} for a three qubit GHZ state (a) and a two qubit Werner state (b) with $p =0.75$.}
\label{fig:trialstates}
\end{figure}

The value of $f(k)$ can be estimated by measuring each of the moments of $\rho^\Gamma$ up to its rank using the circuit in Figure~\ref{fig:momentestimation}. In Figure~\ref{fig:GHZ}, we show the calculation of $f(k)$ for the GHZ state $\frac{1}{\sqrt{2}}(\ket{000}+\ket{111})$, which is known to be maximally entangled. In Figure~\ref{fig:Werner}, we plot $f(k)$ for the Werner state $\rho_w (p) = p \ket{\Psi^-}\!\bra{\Psi^-} + (1-p)\mathbb{I}$ with $p=0.75$, a regime where this state is entangled. In both cases, the value of $f(k)$ is clearly less than zero for multiple values of $k$.

\begin{figure}
\begin{center}
\includegraphics[width=0.5\columnwidth]{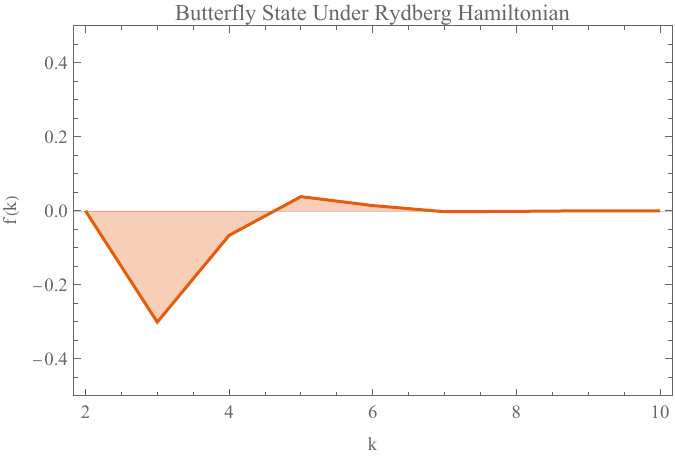}
\end{center}
\caption{A calculation of Eq.~\eqref{eq:closedform} for a ``butterfly metrology state" as given in \cite{kobrin2024universal} obtained from evolution under the Rydberg XY Hamiltonian. }
\label{fig:butterfly}
\end{figure}

Note that $f(k)$ being negative for some $1 < k \leq d$, where $d$ is the rank of the density matrix, does not imply that $f(k+1)$ is also negative. Thus, it does not suffice to simply check $f(d)$ to ascertain whether $\rho$ satisfied the PPT criterion. In Figure~\ref{fig:butterfly}, we give an example where $f(k) < 0$ for $2< k < 5$, but positive again for $k=5$.  The state used in this example is modeled on a four qubit version of the butterfly metrology state from \cite{kobrin2024universal}, which has the form
\begin{equation}
    \ket{\psi (t)} = e^{\frac{i H t}{\hbar}} e^{\frac{i V \pi}{4}} e^{\frac{-i H t}{\hbar}} \ket{\psi(0)}\, .
\end{equation}
For this particular example, we chose $H$ to be the Hamiltonian governing XY interactions in Rydberg atoms in a lattice
\begin{equation}
    H = - J \sum_{i<j}\frac{a^3}{r_{ij}^3}(X_iX_j + Y_iY_j) \, ,
\end{equation} where $J$ and $a$ correspond to the dipole strength and lattice spacing respectively. The initial state was $\ket{\psi(0)}=\ket{1011}$, and the local operator is given by $V = Y_1Y_2$. This example demonstrates the necessity of computing all $d$ inequalities before concluding that $\rho$ has PPT. This aligns with previous work stating that the number of moments estimated should equal the rank of the density matrix \cite{shin2025trace}.

\section{Graph Theoretic Equivalent Conditions}\label{sec:graphs}

A density matrix $\rho$ induces a weighted graph by letting the graph's adjacency matrix be $\rho$ itself, in the sense that the elements $\rho_{ij}$ specify the weights of the edges between vertices $i$ and $j$.
In \cite{bradshaw2023}, it was shown that the cycle index polynomial of the symmetric group evaluated at the moments of $\rho$ appears as the coefficients in the exponential expansion of the graph zeta function $\zeta_\rho$ defined by
\begin{equation}
    \zeta_\rho(u)=\prod_{[P]}(1-N_E(P)u^{\nu(P)})^{-1},
\end{equation}
where the product is over all equivalence classes of prime paths in the induced graph, $\nu(P)$ is the length of the path $P$, and $N_E(P)$ is the product of the weights of the edges in the path $P$. A prime in the weighted graph is a path specified by a sequence $e_1\dots e_k$ of edges which is closed, backtrackless, tailless, and not the power of another path i.e. the origin vertex of $e_1$ is the terminal vertex of $e_k$, we have $e_k\ne e_1^{-1}$, $e_{j+1}\ne e_j^{-1}$ for all $j=1,\ldots, k-1$, and there is no $m$ such that $e_1\cdots e_k=(e_1\cdots e_m)^{k/m}$. The equivalence classes of primes are obtained by identifying cyclic shifts of primes so that $[e_1\cdots e_k]=[e_ke_1\cdots e_{k-1}]$, for example.

Explicitly, it was shown that
\begin{equation}\label{eq:exponential-zeta}
    \zeta_\rho(u)=\exp\left(\sum_{k=1}^\infty\frac{\tr[\rho^k]}{k}u^k\right)=\sum_{k=0}^\infty \frac{Z(S_k)(\tr[\rho],\ldots,\tr[\rho^k])}{k!}u^k,
\end{equation}
and these findings held more generally for any weighted graph with adjacency matrix $\rho$, not just those induced by density matrices.
Thus, we may replace $\rho$ by $\rho^\Gamma$ and manipulate \eqref{eq:exponential-zeta} to obtain
\begin{equation}
    \sum_{k=0}^\infty (-1)^k\frac{Z(S_k)(-\tr[\rho^\Gamma],\ldots,-\tr[(\rho^\Gamma)^k])}{k!}u^k=\exp\left(-\sum_{k=1}^\infty(-1)^k\frac{\tr[(\rho^\Gamma)^k]}{k}u^k\right)=(\zeta_{\rho^\Gamma}(-u))^{-1},
\end{equation}
so that
\begin{equation}
    (\zeta_{\rho^\Gamma}(-u))^{-1} = \sum_{k=0}^\infty \frac{f(k)}{k!}u^k.
\end{equation}
Just as $\zeta_\rho(u)$ is the generating function for the acceptance probabilities of the entanglement tests appearing in \cite{bradshaw2022}, $(\zeta_{\rho^\Gamma}(-u))^{-1}$ is the generating function for the sequence $f(k)$ defining the moment inequalities
which are collectively equivalent to the PPT criterion. Thus, we have proved the following theorem, producing a graph theoretic perspective for the PPT criterion.

\begin{theorem}
    Let $\rho^\Gamma$ be a partial transpose for a $d\times d$ density matrix $\rho$. If $\rho$ is separable, then
    \begin{equation}\label{eq:graph-condition}
        \frac{1}{k!}\frac{d^k}{du^k}\bigg\vert_{u=0}\left(\prod_{[P]}\left(1-(-1)^{\nu(P)}N_E(P)u^{\nu(P)}\right)\right)\ge0
    \end{equation}
    for all $k=1,\ldots,d$. Moreover, \eqref{eq:graph-condition} holds for all $k=1,\ldots,d$ if and only if $\rho$ satisfies the PPT criterion.
\end{theorem}

Let us examine the first several such inequalities. For $k=1$, we have
\begin{equation}\label{eq:k1inequal}
    \sum_{\substack{[P]\\\nu(P)=1}}N_E(P)\ge0.
\end{equation}
Of course, the partial transpose operation does not affect the diagonal of our density matrix, so the weights of the primes of length 1 are unchanged by the transformation $\rho\mapsto\rho^\Gamma$. Thus, the sum in \eqref{eq:k1inequal} is unchanged, and moreover, the sum of the weights of the primes of length 1 is the trace of $\rho$, showing that \eqref{eq:k1inequal} is always trivially satisfied. Turning now to $k=2$, we have
\begin{equation}
    \sum_{\substack{[P]\\\nu(P)=2}}N_E(P)\le\left(\sum_{\substack{[P]\\\nu(P)=1}}N_E(P)\right)^2.
\end{equation}
It was just shown in the $k=1$ case that the right hand side is 1. Thus, the inequality reduces to 
\begin{equation}\label{eq:k2inequal}
    \sum_{\substack{[P]\\\nu(P)=2}}N_E(P)\le1,
\end{equation}
which is trivially satisfied by a diagonal adjacency matrix. Since $\rho^\Gamma$ is self-adjoint, it can always be diagonalized, and so \eqref{eq:k2inequal} too holds trivially. The $k=3$ inequality is more interesting. We have
\begin{align}
    \sum_{\substack{[P]\\\nu(P)=3}}N_E(P)&\ge\frac{1}{6}\left[5\sum_{\substack{[P]\\\nu(P)=1}}N_E(P)\sum_{\substack{[P]\\\nu(P)=2}}N_E(P)-\left(\sum_{\substack{[P]\\\nu(P)=1}}N_E(P)\right)^3\right]\\
    &=\frac{1}{6}\left[5\sum_{\substack{[P]\\\nu(P)=2}}N_E(P)-1\right]=\frac{1}{24}.
\end{align}
Take, for example, the Bell state $\frac{1}{\sqrt{2}}(\ket{00}+\ket{11})$, which has density matrix
\begin{equation}
    \rho = \begin{pmatrix}
        1/2 & 0 & 0 & 1/2\\
        0   & 0 & 0 & 0\\
        0   & 0 & 0 & 0\\
        1/2 & 0 & 0 & 1/2
    \end{pmatrix}.
\end{equation}
\begin{figure}
\begin{center}
\includegraphics[width=0.5\columnwidth]{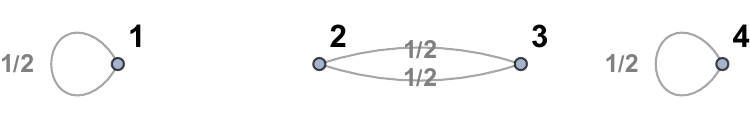}
\end{center}
\caption{Graph induced by the partial transpose of the Bell state $\frac{1}{\sqrt{2}}(\ket{00}+\ket{11})$.}
\label{fig:bell}
\end{figure}The partial transpose with respect to the second subsystem is given by transposing the four $2\times 2$ blocks. Thus, we have
\begin{equation}
    \rho = \begin{pmatrix}
        1/2 & 0 & 0 & 0\\
        0   & 0 & 1/2 & 0\\
        0   & 1/2 & 0 & 0\\
        0 & 0 & 0 & 1/2
    \end{pmatrix}.
\end{equation}
The induced graph is shown in Figure~\ref{fig:bell} and consists of four vertices labeled $1,2,3,4$, and four edges, an edge $e_1$ from vertex 1 to itself, an edge $e_{23}$ from vertex 2 to vertex 3, an edge $e_{32}$ from vertex 3 to vertex 2, and an edge $e_4$ from vertex 4 to itself. The weight of every such edge is $1/2$. The only equivalence class of prime paths with length $\nu(P)=2$ is given by the representative $P=e_{23}e_{32}$ which has edge norm $N_E(P)=\frac{1}{4}$. Moreover, there are no equivalence classes of prime paths with length 3. Thus, 
\begin{equation}
    \sum_{\substack{[P]\\\nu(P)=3}}N_E(P)=0<\frac{1}{24}=\frac{1}{6}\left[5\sum_{\substack{[P]\\\nu(P)=2}}N_E(P)-1\right],
\end{equation}
and it follows that the Bell state is entangled.

\section{Conclusion}\label{sec:conclusion}

In this work, we have examined the relationships between the moments of the partial transpose of a separable state that arise from Newton's identities and Descartes' rule of sign. We have augmented the original approach of Neven et al. \cite{neven2021} with a closed form expression for the associated moment-based entanglement tests. Explicitly, we have given a closed form for the relationship between a state $\rho$ having PPT and the elementary symmetric polynomials, showing that these polynomials specify a series of inequalities that must be obeyed to satisfy the PPT criterion. Finally, as a point of interest, we share an equivalent graph-theoretic interpretation of these inequalities. 

While this derivation is mathematically interesting, the relationship between entanglement witnesses such as the PPT criterion and symmetric polynomials is perhaps more so. Previously, cycle index polynomials of various groups were shown to correspond to the acceptance probability of purity tests, which can also act as entanglement witnesses. A potential avenue for future work is determining if further linear entanglement witnesses give rise to similar moment-based entanglement tests which can be written in terms of symmetric polynomials in their closed forms. However, analytic forms of these expression are not always available and can additionally be an area of future research.

\section*{Acknowledgments}

The authors acknowledge support from the Office of the Under Secretary of Defense for Research and Engineering (OUSD(R\&E)) SMART SEED award. 

\section*{Competing Interests} The authors declare no non-financial competing interests. This research was funded by the Office of the Under Secretary of Defense for Research and Engineering (OUSD(R\&E)) SMART SEED award.

\section*{Data Availability Statement} There is no supporting data to accompany this work.
\bibliographystyle{unsrt}
\bibliography{ref}

\end{document}